\documentclass[11pt]{article}

\usepackage{epsf,epsfig,graphicx}
\graphicspath{{figures/}}

\usepackage{mathrsfs}
\usepackage{amsmath,amssymb,amsthm}
\usepackage{enumerate}
\usepackage{stackrel}

 \newtheorem{theorem}{Theorem}[section]
 \newtheorem{lemma}{Lemma}[section]
 
 \newtheorem{definition}{Definition}[section]

\newcommand{\OPT}{{\textnormal{OPT}}}

\theoremstyle{definition}
\newtheorem{algz}{Algorithm}[section]
\theoremstyle{plain}
\newcounter{algleo}[algz]
\newlength{\lefttab}
\newlength{\numberoffset}
\newenvironment{algleo}%
  {\trivlist
   \topsep=0pt\parsep=0pt\itemsep=0pt
   \def\li{\item\refstepcounter{algleo}\makebox[0.8em][r]{\thealgleo\hspace{\numberoffset}}
       \hangafter1\hangindent1.8em\noindent}%
   \def\linonumber{\item\makebox[0.8em][r]{\hspace{\numberoffset}}
       \hangafter1\hangindent1.8em\noindent}%
   \def\linooffset{\item\hangafter1\hangindent1em\noindent}%
   \addtolength{\lefttab}{1.25em}
   \addtolength{\numberoffset}{1.25em}
   \leftskip=\lefttab}%
  {\endtrivlist}

\begin{document}

\title{Approximating the Weighted Minimum Label $s$-$t$ Cut Problem}

\author{
Peng Zhang
  \thanks{School of Software,
    Shandong University, Jinan, Shandong, 250101, China.
    E-mail: {\tt algzhang@sdu.edu.cn}.
    ORCID: {\tt 0000-0001-7265-0983}.}
}

\maketitle

\begin{abstract}
In the weighted (minimum) {\sf Label $s$-$t$ Cut} problem, we are given
a (directed or undirected) graph $G=(V,E)$,
a label set $L = \{\ell_1, \ell_2, \dots, \ell_q \}$
with positive label weights $\{w_\ell\}$, a source $s \in V$ and
a sink $t \in V$.
Each edge edge $e$ of $G$ has a label $\ell(e)$ from $L$.
Different edges may have the same label.
The problem asks to find a minimum weight label subset $L'$
such that the removal of all edges with labels in $L'$ disconnects $s$ and $t$.

The unweighted {\sf Label $s$-$t$ Cut} problem
(i.e., every label has a unit weight) can be approximated within $O(n^{2/3})$,
where $n$ is the number of vertices of graph $G$.
However, it is unknown for a long time how to approximate the weighted
{\sf Label $s$-$t$ Cut} problem within $o(n)$. In this paper, we provide
an approximation algorithm for the weighted {\sf Label $s$-$t$ Cut}
problem with ratio $O(n^{2/3})$. The key point of the algorithm is a mechanism
to interpret label weight on an edge as both its length and capacity.
\end{abstract}

%

\section{Introduction}
The (minimum) {\sf Label $s$-$t$ Cut} problem
is a natural combinatorial optimization problem arising
from several application fields, including system security
\cite{JSW02,SHJ+02,SW04}, computer networks \cite{CDP+07,CPRV16},
and control engineering \cite{RKK06}.
The problem is defined as in Definition \ref{def - label s-t cut}.
{\sf Label $s$-$t$ Cut} has attracted much attention from researchers
in computer science
(see, e.g., \cite{BCM+16,CJPP17,CPRV16,FGK10,GKP17,GMT15,GCSR13,JB17})
and researchers even in chemical engineering (see \cite{FWXG17}).
The problem finds many applications in image segmentation \cite{KOJ13},
network connectivity \cite{ZM18}, computational biology \cite{OSC11},
and network security \cite{BS17c}, etc.

\begin{definition}
\label{def - label s-t cut}
The minimum {\sf Label $s$-$t$ Cut} problem.

{\em (Instance)} We are given a (directed or undirected) graph $G=(V,E)$,
a label set $L = \{\ell_1, \ell_2, \dots, \ell_q \}$,
a source $s \in V$ and a sink $t \in V$.
Each edge edge $e$ of $G$ has a label $\ell(e)$ from $L$.
Different edges may have the same label.

{\em (Goal)} The problem asks to find a minimum size label subset $L'$
such that the removal of all edges with labels in $L'$ disconnects $s$ and $t$.
\end{definition}

In the weighted {\sf Label $s$-$t$ Cut} problem, a positive label weight
$w_\ell$ is provided further for every label $\ell \in L$.
The problem asks to find a minimum {\em weight} label subset $L'$
to disconnect $s$ and $t$, where the weight $w(L)'$ of label subset $L'$
is defined as $w(L') = \sum_{\ell \in L} w_\ell$.

We briefly introduce some origins and applications of the {\sf Label $s$-$t$ Cut}
problem. The detailed description can be found in \cite{CDP+07}, \cite{JSW02},
and related references therein.
{\sf Label $s$-$t$ Cut} comes from an interesting scenario in system security.
An intruder with some attack methods lies in his beginning state,
planing intrude a system. The intruder's current state changes once he applies
one of the attack methods. The successful state means that the intruder has
intruded the system successfully. The state transition of the intruder
can be modeled by an edge-labeled graph with source $s$ (representing
the beginning state) and sink $t$ (representing the successful state).
The system defender can disable the intruder's attack methods by paying some costs,
e.g., by strengthening the defending equipments.
The system defender's task can be depicted as finding the minimum cost of labels
such that $s$ and $t$ are disconnected in the graph after removing the edges
with these labels. This is precisely the {\sf Label $s$-$t$ Cut} problem.

As an application scenario, let us consider computer networks, which are usually
multi-layered. For simplicity, we assume that there are two layers in a network,
namely, the high level logical layer and the low level physical layer.
Each link in the logical layer corresponds to a path in the physical layer.
Since many logical links may rely on paths that have some common physical link,
a failure of some physical link may result in a failure of many seemingly
unrelated logical links. This can be modeled by an edge-labeled graph.
More importantly, {\sf Label $s$-$t$ Cut} can be used to compute
the generalized $s$-$t$ connectivity (counts in labels) in edge-labeled graphs.

\medskip

{\bf Notation.}
Let $G$ be a graph in the context. We always use $n$ to denote the number
of vertices of $G$, and use $m$ to denote the number of edges of $G$.
Moreover, we use $\OPT$ to denote the optimal solution value of an optimization
problem. Note that for the {\sf Label $s$-$t$ Cut} problem, we still have
a natural parameter $q$, which is the number of labels in the problem input.
For simplicity, we will use $\ell$ to denote a label in $L$.
Let $e$ be an edge. We also use $\ell(e)$ to denote the label on edge $e$.
Note that in the latter case $\ell$ is actually a function from $E$ to $L$.
The two cases of usage of $\ell$ will be easily distinguished from the context.
The same thing happens to symbol $L$. On the one hand, $L$ is the set of labels
in the problem input. On the other hand, given an edge subset $E'$, we will
use $L(E')$ to denote the set of labels appearing on edges in $E'$.

\subsection{Related Work}
It is easy to see that {\sf Label $s$-$t$ Cut} is just a generalization
of the classic {\sf Min $s$-$t$ Cut} problem. If in {\sf Label $s$-$t$ Cut}
every edge has a unique label, then {\sf Label $s$-$t$ Cut} degenerates
to {\sf Min $s$-$t$ Cut}. While it is well-known that {\sf Min $s$-$t$ Cut}
is polynomial-time solvable, {\sf Label $s$-$t$ Cut} is NP-hard
and has high approximation hardness factor $2^{(\log n)^{1-1/(\log\log n)^c}}$
for any constant $c < 1/2$ \cite{ZCTZ11}.
For the applications of {\sf Label $s$-$t$ Cut} in system security
\cite{JSW02,SHJ+02,SW04} and computer networks \cite{CDP+07,CPRV16},
the readers are advised to refer to the corresponding references.

The first non-trivial approximation algorithm for the weighted
{\sf Label $s$-$t$ Cut} problem was given by Zhang et al. \cite{ZCTZ11}
in 2011 with performance ratio $O(m^{1/2})$, where $m$ is the number of edges
in the input graph. However, in dense graphs with $m = \Theta(n^2)$,
the approximation ratio $O(m^{1/2})$ degenerates to $O(n)$.
Therefore, an approximation algorithm with ratio in terms of $n$
(better than $O(n)$) is expected for weighted {\sf Label $s$-$t$ Cut}.

Rohloff et al. \cite{RKK06} gave an $O(n^{2/3})$-approximation algorithm
for unweighted {\sf Label $s$-$t$ Cut}.
Tang et al. \cite{TZ12} independently gave
an $O(\frac{n^{2/3}}{\OPT^{1/3}})$-approximation algorithm
for the unweighted {\sf Label $s$-$t$ Cut} problem
via a mixed strategy of LP-rounding and finding a min $s$-$t$ cut,
where $\OPT$ denotes the value of an optimal solution to the problem.
Thereafter, Zhang et al. \cite{ZFT18} rendered the algorithm in \cite{TZ12}
purely combinatorial, keeping the approximation ratio $O(\frac{n^{2/3}}{\OPT^{1/3}})$
unchanged. Dutta et al. \cite{DHKM16} also gave an $O(n^{2/3})$-approximation
algorithm for unweighted {\sf Label $s$-$t$ Cut}.

However, all the algorithms in \cite{RKK06}, \cite{TZ12}, \cite{ZFT18},
and \cite{DHKM16} only deal with the {\em unweighted} {\sf Label $s$-$t$ Cut} problem.
To the best of our knowledge, there is no approximation algorithm
for the weighted {\sf Label $s$-$t$ Cut} problem with performance ratio
directly given in terms of $n$ (better than $O(n)$).
Note that for an optimization problem, the complexity of the unweighted version
and the weighted version could be very different. For example,
the well-known {\sf $k$-MST} problem in weighted graphs is NP-hard \cite{RSM+96}.
However, in unweighted graphs {\sf $k$-MST} is polynomial time solvable.

\subsection{Our Results}
In this paper, we provide an approximation algorithm for the weighted
{\sf Label $s$-$t$ Cut} problem with ratio $O(n^{2/3})$. This is the first
approximation algorithm for weighted {\sf Label $s$-$t$ Cut} whose ratio
is given in terms of $n$. The key point of the algorithm is a mechanism
to interpret label weight on an edge as both the edge's length and capacity.

The overall strategy is to discretize label weights $w_\ell$ with arbitrary values
into integer values ${\bar w}_\ell$ in a range from one to an upper bound
polynomial in $q$.
Recall that $q$ is the number of labels in the problem input.
Every edge $e$ in the input graph $G$ is transformed into $\bar{w}_{\ell(e)}$
parallel edges, resulting in a multi-graph $\tilde G$.
Meanwhile, every label $\ell \in L$ is replaced
by a group of labels $\ell^{(1)}, \ell^{(2)}, \dots, \ell^{(\bar{w}_\ell)}$,
resulting in a new label set $\tilde{L}$. In this way, we get
a new {\sf Label $s$-$t$ Cut} instance $\tilde{\cal I}$ on multi-graphs.
Since $\bar{w}_{\ell}$ is polynomially bounded, instance $\tilde{\cal I}$
can be constructed in polynomial time.
Next, we develop a two-stage combinatorial algorithm
for {\sf Label $s$-$t$ Cut} on multi-graphs to compute an approximate
solution $\tilde{L}' \subseteq \tilde{L}$ to instance $\tilde{\cal I}$.
Finally, a solution $L' \subseteq L$ to the original instance
is recovered from $\tilde{L}'$.

\subsection{More Related Work}
The {\sf Label $s$-$t$ Cut} problem was also studied from the parameterized
perspective. Fellows et al. \cite{FGK10} showed that
when parameterized by the number of used labels,
the {\sf Label $s$-$t$ Cut} problem is $W[2]$-hard even in graphs whose
path-width is bounded above by a small constant.
On the other hand, Zhang et al. \cite{ZF16} showed that when parameterized
by the number of used labels, {\sf Label $s$-$t$ Cut} is fixed-parameter tractable
if the maximum length of any $s$-$t$ path is bounded from above.
Morawietz et al. \cite{MGKS20} considered the parameterized complexity of
{\sf Label $s$-$t$ Cut} with parameters related to the structure of
the input graph and labels.

Jegelka et al. \cite{JB17} studied a more general cut problem called
{\sf Cooperative $s$-$t$ Cut}, which finds an $s$-$t$ cut such that
an objective function is minimized, where the objective function
can be arbitrary submodular function defined on $2^E$.
It is not difficult to see that {\sf Cooperative $s$-$t$ Cut} is
a generalization of {\sf Label $s$-$t$ Cut}.
Jegelka et al. \cite{JB17} proved {\sf Cooperative $s$-$t$ Cut}
cannot be approximated within $\sqrt{n \log n}$.
It is interesting to note that this lower bound is based on
information theory \cite{SF11}, instead of complexity assumptions
such as $\text{P} \neq \text{NP}$.

We would like to introduce here the {\sf Global Label Cut} problem
\cite{Far06,CDP+07,TZ12,GKP17},
which is closely related to {\sf Label $s$-$t$ Cut}.
In the {\sf Global Label Cut} problem, there are no given vertices $s$ and $t$.
The problem asks to find a minimum size (or weight) label subset
$L' \subseteq L$ such that the removal of edges whose labels are in $L'$
makes the undirected input graph disconnected.
It is easy to see that {\sf Global Label Cut} is a natural generalization
of the classic global {\sf Min Cut} problem (see, e.g., \cite{KS96}
for {\sf Min Cut}).

To the best of our knowledge, {\sf Global Label Cut} had been proposed
at least in 2006 by Farag{\'o} \cite{Far06}.
Zhang et al. \cite{ZF16} showed that in several special cases,
{\sf Global Label Cut} is polynomial-time solvable.
An important progress was made by Ghaffari et al. \cite{GKP17},
who gave a quasi-polynomial time Monte-Carlo algorithm for
unweighted {\sf Global Label Cut}, and a PTAS with high probability
for weighted {\sf Global Label Cut}.
Bordini et al. \cite{BPGF19} gave some heuristics for {\sf Global Label Cut}.
At the current time, the most intriguing open problem
is whether {\sf Global Label Cut} is in P or NP-hard.

\bigskip

{\bf Organization of the paper.}
The remainder of the paper is organized as follows.
In Section \ref{sec - unweighted LstC}, we give a simple analysis
to the existing algorithm for the unweighted {\sf Label $s$-$t$ Cut} problem.
In Section \ref{sec - high-level idea}, we exhibit our high-level idea
about how to approximate weighted {\sf Label $s$-$t$ Cut}.
In Section \ref{sec - LstC on multi-graphs with forbidden labels}, we show
how to approximate unweighted {\sf Label $s$-$t$ Cut} on multi-graphs.
Then, in Section \ref{sec - weighted LstC} we show how to reduce
the weighted {\sf Label $s$-$t$ Cut} problem to the unweighted one on multi-graphs.
Finally, in Section \ref{sec - conclusions} we conclude the paper.

\section{Unweighted Label $s$-$t$ Cut}
\label{sec - unweighted LstC}
It is known that unweighted {\sf Label $s$-$t$ Cut} can be approximated
within $O(\frac{n^{2/3}}{\OPT^{1/3}})$ \cite{ZFT18}. In this section,
we would like to give a simple analysis to this result.

It is known that {\sf Label $s$-$t$ Cut} on unweighted graphs can be
approximated within $O(\frac{n^{2/3}}{\OPT^{1/3}})$ \cite{TZ12,ZFT18}.
In this paper, we give a new, but shorter proof for this result.
To the best of our knowledge, the new proof should be the shortest one
known so far for the $O(\frac{n^{2/3}}{\OPT^{1/3}})$-approximation result.

In 2006, Rohloff et al. \cite{RKK06} gave an $O(n^{2/3})$-approximation
algorithm for unweighted {\sf Label $s$-$t$ Cut}.
The paper \cite{RKK06} deals with the sensor selection problem arising
in the field of control theory and engineering.
The algorithm for {\sf Label $s$-$t$ Cut} in \cite{RKK06} has two stages,
with the first stage repeatedly finding the shortest $s$-$t$ paths,
and the second stage running a depth-first search.

In 2012, Tang et al. \cite{TZ12} independently
gave an $O(\frac{n^{2/3}}{\OPT^{1/3}})$-approximation
algorithm for the unweighted {\sf Label $s$-$t$ Cut} problem.
The algorithm in \cite{TZ12} is also a two-stage algorithm,
with the first stage rounding a fractional solution to
the following linear program (\ref{LP - LP1 for LstC}),
and the second stage finding a min $s$-$t$ cut.

\begin{alignat}{2}
\label{LP - LP1 for LstC}
\min \quad & \sum_{\ell \in L} x_\ell  &{}& \tag{LP1} \\
\mbox{s.t.} \quad
& \nonumber
\sum_{e \in P} x_{\ell(e)} \geq 1,
    &\qquad& \forall \text{$s$-$t$ path~} P \\
\nonumber
& x_\ell \geq 0,
    &{}& \forall \ell \in L
\end{alignat}

Thereafter, Zhang et al. \cite{ZFT18} (in 2018)
rendered the algorithm in \cite{TZ12} purely combinatorial,
keeping the approximation ratio $O(\frac{n^{2/3}}{\OPT^{1/3}})$ unchanged.
More specifically, in the first stage of the algorithm in \cite{ZFT18},
LP-rounding is no longer needed, but a work of repeatedly finding
the shortest $s$-$t$ paths is placed instead.

In 2016, Dutta et al. \cite{DHKM16} also independently
gave an $O(n^{2/3})$-approximation algorithm for {\sf Label $s$-$t$ Cut}.
Their algorithm is still a two-stage
algorithm: LP-rounding and finding a min $s$-$t$ cut.
The linear program used in \cite{DHKM16} is
the following linear program (\ref{LP - LP2 for LstC}),
which is different to that in \cite{TZ12}.
\begin{alignat}{2}
\label{LP - LP2 for LstC}
\min \quad & \sum_{\ell \in L} x_\ell  &{}& \tag{LP2} \\
\mbox{s.t.} \quad
& \nonumber
\sum_{\ell \in L(P)} x_{\ell} \geq 1,
    &\qquad& \forall \text{$s$-$t$ path~} P \\
\nonumber
& x_\ell \geq 0,
    &{}& \forall \ell \in L
\end{alignat}
Since the separation problem to the LP in \cite{DHKM16} is NP-hard,
it is not known how to get an optimal solution to this LP in polynomial time.
Dutta et al. \cite{DHKM16} did much work to get an approximate solution
to this LP. However, the LP in \cite{TZ12} actually is enough for the algorithm
in \cite{DHKM16}.

It has mentioned that all the algorithms in \cite{RKK06}, \cite{TZ12}
(including \cite{ZFT18}), and \cite{DHKM16} are two-stage algorithms.
However, the analyses of the second stages in them are different
from each other. Inspired by the analysis of the second stage in \cite{DHKM16},
we give a more simpler analysis to the algorithm in \cite{ZFT18}.
To the best of our knowledge, this should be the simplest analysis to
this algorithm known so far. The algorithm is shown
as Algorithm \ref{alg - unweighted LstC - n}, which is an approximation
algorithm for the unweighted {\sf Label $s$-$t$ Cut} problem
on (directed or undirected) simple graphs.

\begin{algz}[\cite{ZFT18}]
\label{alg - unweighted LstC - n}
\setcounter{algleo}{0}
\begin{algleo}
\linonumber
\linooffset {\em Input:} An instance ${\cal I} = (G, s, t, L)$ of
    {\sf Label $s$-$t$ Cut}.
\linooffset {\em Output:} A label subset of $L$.
\li \label{step - guess OPT}
    Guess $\OPT$.
\linonumber /* Stage one */
\li $L_1 \leftarrow \emptyset$.
\li {\bf while} {the $s$-$t$ distance
    $\leq \frac{n^{2/3}}{\OPT^{1/3}}$} {\bf do}
    \begin{algleo}
    \li Find a shortest $s$-$t$ path $P$.
    \li Remove all the edges whose labels are in $L(P)$.
    \li $L_1 \leftarrow L_1 \cup L(P)$.
    \end{algleo}
\li {\bf endwhile}
\linonumber /* Stage two */
\li Denote by $R$ the current remaining graph.
    Find a minimum size $s$-$t$ cut $E' \subseteq E(R)$ of $R$.
\li $L_2 \leftarrow L(E')$.
\li {\bf return} $L_1 \cup L_2$.
\end{algleo}
\end{algz}

\begin{theorem}[\cite{GY95}]
\label{th - Short length versions of Menger's theorem}
For a simple $n$-vertex (undirected or directed) graph $G$,
if there are $k$ edge-disjoint paths between two vertices $s$ and $t$,
then the average length (number of edges) of these paths is $O(n/\sqrt{k})$.
\end{theorem}

\begin{theorem}
\label{th - unweighted LstC - n}
Algorithm \ref{alg - unweighted LstC - n} is
an $O(\frac{n^{2/3}}{\OPT^{1/3}})$-approximation algorithm
for unweighted {\sf Label $s$-$t$ Cut} on (directed or undirected) simple graphs.
\end{theorem}
\begin{proof}
First note that the algorithm obviously gives a feasible solution in polynomial time.
Every path found in stage one has length $\leq \frac{n^{2/3}}{\OPT^{1/3}}$.
All these paths are label-disjoint (i.e., no two paths of them contain the same label),
and hence edge-disjoint. It follows that the optimal solution
must select at least one label from each path,
implying that the number of these path is $\leq \OPT$.
Therefore, we get that
\begin{equation}
\label{eqn - L1 for unweighted graph}
|L_1| \leq \frac{n^{2/3}}{\OPT^{1/3}} \OPT.
\end{equation}

By Menger's theorem (see, e.g., \cite{KV12}), the size (number of edges) $|E'|$ of
the minimum $s$-$t$ cut is equal to the number edge disjoint $s$-$t$ paths
in graph $R$.
By Theorem \ref{th - Short length versions of Menger's theorem},
the average length of these paths is $O(n/\sqrt{|E'|})$.
Let $d_R(s, t)$ be the length of the shortest $s$-$t$ path in graph $R$.
When Algorithm \ref{alg - unweighted LstC - n} enters stage two,
it holds that $d_R(s, t) \geq \frac{n^{2/3}}{\OPT^{1/3}}$. We thus have
$\frac{n^{2/3}}{\OPT^{1/3}} \leq d_R(s, t) \leq O(n/\sqrt{|E'|})$,
implying that
\begin{equation}
\label{eqn - L2 for unweighted graph}
|L_2| \leq |E'| = O(n^{2/3}\OPT^{2/3}) = O\left(\frac{n^{2/3}}{\OPT^{1/3}}\right) \OPT.
\end{equation}

The theorem then follows by (\ref{eqn - L1 for unweighted graph}) and
(\ref{eqn - L2 for unweighted graph}).
\end{proof}

The guess technique in step \ref{step - guess OPT} of
Algorithm \ref{alg - unweighted LstC - n} is a common technique
in approximation algorithms. By ``guess $\OPT$'' we actually mean
repeating the algorithm for each possible value of $\OPT$ in $\{1, 2, \dots, q\}$,
and returning the solution of minimum size ever found.
The range of $\OPT$ is polynomially bounded, so the whole algorithm
still runs in polynomial time.

\section{High-level Idea of Approximating Weighted Label $s$-$t$ Cut}
\label{sec - high-level idea}
Note that in the weighted {\sf Label $s$-$t$ Cut} problem weights are defined on labels
(instead of edges). Our overall strategy dealing with weighted {\sf Label $s$-$t$ Cut},
as the case of unweighted {\sf Label $s$-$t$ Cut}, is still a two-stage algorithm.
In the first stage, we repeatedly find shortest $s$-$t$ paths and remove them.
In the second stage, we compute an $s$-$t$ cut.
However, the problem we are faced is a weighted problem, and we have to consider
the label weights on edges when we compute $s$-$t$ path and $s$-$t$ cut.
This immediately brings out a troublesome problem: For $s$-$t$ path,
label weight on an edge has to be interpreted as edge length.
However, for $s$-$t$ cut, label weight has to be interpreted as edge capacity.
How could we combine these two very different physical concepts together
in an algorithm?

Our goal is to give an approximation algorithm for weighted {\sf Label $s$-$t$ Cut}
whose ratio is given in terms of $n$ (the number of vertices of the input graph),
since there is already an $O(m^{1/2})$-approximation algorithm \cite{ZCTZ11}
for the problem. A straightforward approach is to use label weight $w_{\ell(e)}$,
for an edge $e$, as both its length and capacity in
Algorithm \ref{alg - unweighted LstC - n}. However, the proof of
Theorem \ref{th - unweighted LstC - n} does not applies to the new algorithm.
There is a result corresponding to
Theorem \ref{th - Short length versions of Menger's theorem}
for edge-weighted graph by Karger and Levine \cite[Theorem 7.1]{KL98}.
However, it also seems that this result does not help for the analysis.

Our strategy is indirect and needs some transformations.
Let us examine Algorithm \ref{alg - unweighted LstC - n} again.
This algorithm deals with an unweighted graph $G$.
Note that an unweighted edge can be viewed as a unit-weighted one.
This means unit-weight works well for not only the first stage of computing
shortest paths, but also the second stage of computing a min $s$-$t$ cut.
In other words, length and capacity ``agree'' on unit weight,
although they differ very much on arbitrary weights.

Based on the above observation, our general strategy of approximating weighted
{\sf Label $s$-$t$ Cut} is to reduce the weighted problem to the unweighted
{\sf Label $s$-$t$ Cut} problem on multi-graphs. Then, we design a two-stage
approximation algorithm for the latter problem. In the reduction we will discretize
the label weights to edge multiplicities. To guarantee that the reduction can
be performed in polynomial time, we distinguish light label weights from
heavy label weights. We guess a weight threshold $W$ (which is the maximum
label weight used in an optimal solution to weighted {\sf Label $s$-$t$ Cut}).
The label weights that are at most $W$ are called {\em light}.
Otherwise, they are called {\em heavy}.
Let $e$ be an edge with light label weight $w_{\ell(e)}$.
We round the weight $w_{\ell(e)}$ into a multiple (say ${\bar w}_{\ell(e)}$)
of some small quantity and represent $e$ by ${\bar w}_{\ell(e)}$ multi-edges,
so that the multiplicity is polynomially bounded.
On the other hand, if $e$'s label weight is heavy, we do not do the conversion.
In this way, we get a multi-graph, denoted by $\tilde G$.
Then we solve unweighted {\sf Label $s$-$t$ Cut} on multi-graph $\tilde G$.
The solution obtained will be converted back to a solution to weighted
{\sf Label $s$-$t$ Cut}.

Our method solving unweighted {\sf Label $s$-$t$ Cut} on multi-graphs
is a two-stage algorithm consisting of computing shortest $s$-$t$ paths
and $s$-$t$ cut. However, we do not want to use the edges
if their original label weights are heavy. To this aim, for an edge $e$
in multi-graph $\tilde G$, if it comes from an edge in $G$ with light
label weight, its weight is defined as one.
Otherwise (i.e., it comes from an edge in $G$ with
heavy label weight), its length is defined as zero. In this way,
we further turn $\tilde G$ to a multi-graph with edge weights zero or one.
Note that zero weight plays an important role in the algorithm.
When we measure the length of an $s$-$t$ path in the first stage,
zero weight edges are not taken into count.
More importantly, in the second stage we use a layering technique to compute
an $s$-$t$ cut. The layering technique guarantees that length-zero edges
in $\tilde G$ will not be included in the $s$-$t$ cut to be found.

To make the idea clear in conception, for an edge $e$
in multi-graph $\tilde G$, if it comes from an edge in $G$ with light
label weight, its label is called {\em admissible}.
Otherwise its label is called {\em forbidden}. So in the above,
we actually define edge weight as one if the edge's label is admissible,
and zero if forbidden.
In Section \ref{sec - LstC on multi-graphs with forbidden labels}
we show in detail the {\sf Label $s$-$t$ Cut} problem on multi-graphs
with forbidden labels, and its two-stage approximation algorithm.

\section{Label $s$-$t$ Cut on Multi-graphs with Forbidden Labels}
\label{sec - LstC on multi-graphs with forbidden labels}
In the {\sf Label $s$-$t$ Cut} problem on multi-graphs with forbidden labels,
we are given a multi-graph $G = (V, E)$, a source $s \in V$ and a sink $t \in V$,
and two label sets $A$ and $B$. Each edge in $E$ has one label in $A \cup B$.
The problem asks to find a minimum size subset $A' \subseteq A$ such that
the removal of edges with labels in $A'$ disconnects $s$ and $t$ in $G$,
or report infeasibility when no solution exists.
The labels in $A$ are called {\em admissible labels}, that is, these labels
can be used in solutions. The labels in $B$ are called {\em forbidden labels},
which cannot be used in any solution.

If an edge's label is an admissible label, then we call this edge
an admissible edge. Likewise, if an edge's label is a forbidden label,
then we call this edge a forbidden edge.

For the {\sf Label $s$-$t$ Cut} problem on multi-graphs with forbidden labels,
we design a two-stage approximation algorithm,
as shown in Algorithm \ref{alg - LstC on multi-graphs with A, B}.
The algorithm first assigns weights to edges according to the types of their labels
(step \ref{step - assign weights to edges}). Then, the algorithm
enters its first stage (step \ref{step - beginning of stage one, MG-AB} to
step \ref{step - end of stage one, MG-AB}), repeatedly finding the current shortest
$s$-$t$ paths and removing selected edges until the $s$-$t$ distance
is relatively long. After that, the algorithm enters its second stage
(step \ref{step - beginning of stage two, MG-AB} to
step \ref{step - end of stage two, MG-AB}), partitioning the vertices into layers
and then finding an $s$-$t$ cut between two consecutive layers.

\begin{algz}
\label{alg - LstC on multi-graphs with A, B}
\setcounter{algleo}{0}
\begin{algleo}
\linonumber
\linooffset {\em Input:} An instance ${\cal I} = (G, s, t, A, B)$ of
    {\sf Label $s$-$t$ Cut} on multi-graphs with forbidden labels.
\linooffset {\em Output:} A label subset $A' \subseteq A$ such that
    the removal of edges with labels in $A'$ disconnects $s$ and $t$,
    if instance $\cal I$ has feasible solutions (i.e., the removal of edges
    with admissible labels can disconnect $s$ and $t$). Returns ``failure''
    otherwise.
\li \label{step - assign weights to edges}
    Define weights (i.e., lengths) on edges of $G$:
    For every $e \in E(G)$, if $\ell(e) \in A$, then define $w(e) = 1$.
    Otherwise (i.e., $\ell(e) \in B$), define $w(e) = 0$.
\linonumber /* The length of a path $P$ is defined
    as the total weight of edges in $P$.
    The distance $\mathrm{dist}(u, v)$ between a pair of vertices $u, v$
    is defined as the length of a shortest $u$-$v$ path.
    If $v$ is not reachable from $u$, then define $\mathrm{dist}(u, v) = \infty$. */
\li \label{step - failure}
    {\bf if} $\mathrm{dist}(s, t) = 0$ {\bf then} {\bf return} ``failure''.
\linonumber /* Stage one */
\li \label{step - beginning of stage one, MG-AB}
    Guess $\OPT$.
    Let $\mu = \mu(G)$ be the maximum multiplicity of edges in $G$.
\li $A_1 \leftarrow \emptyset$.
\li {\bf while} {$\mathrm{dist}(s, t) \leq \frac{n^{2/3}\mu^{1/3}}{\OPT^{1/3}}$}
    {\bf do}
    \begin{algleo}
    \li Find a shortest $s$-$t$ path $P$.
    \li Remove all the edges whose labels are in $L(P) \cap A$.
    \li $A_1 \leftarrow A_1 \cup (L(P) \cap A)$.
    \end{algleo}
\li \label{step - end of stage one, MG-AB}
    {\bf endwhile}
\linonumber /* Stage two */
\li \label{step - beginning of stage two, MG-AB}
    Denote by $R$ the remaining graph.
    Let
    \[
    V_i := \{v \in V(R) \mid \mathrm{dist}(s, v) = i, i \neq \infty \}
    \]
    for $i = 0, 1, 2, \dots$.
    This gives a partition $\{V_0, V_1, V_2, \dots\}$ of the vertices in $R$.
    Suppose that $t \in V_\tau$ for some integer $\tau \geq 1$.
\li Let $E_i$ ($0 \leq i \leq \tau-1$) be the set of edges between
    two consecutive layers $V_i$ and $V_{i+1}$.
    Every $E_i$ forms an $s$-$t$ cut in $R$.
    Choose the minimum cardinality one among these cuts, denoted by $E'$.
\li \label{step - end of stage two, MG-AB}
    $A_2 \leftarrow L(E')$.
\li {\bf return} $A_1 \cup A_2$.
\end{algleo}
\end{algz}

Note that graph $G$ varies during the execution
of Algorithm \ref{alg - LstC on multi-graphs with A, B}.
The distance $\mathrm{dist}(u, v)$ is always the distance between $u$ and $v$
in the current graph.

\begin{lemma}
\label{lm - |A1| on multi-graphs with A, B}
For Algorithm \ref{alg - LstC on multi-graphs with A, B}, we have
$|A_{1}| \leq n^{2/3} (\frac{\mu}{\OPT})^{1/3} \OPT$.
\end{lemma}
\begin{proof}
Let $h$ be the number of paths found in the first stage.
Since these $s$-$t$ paths do not share any admissible label, any optimal solution
must select at least one admissible label from each of these paths.
Therefore, we know $\OPT \geq h$.
By the definition of costs on edges, the number of admissible labels
on each such path is at most its length,
which is in turn at most $\frac{n^{2/3}\mu^{1/3}}{\OPT^{1/3}}$.
Thus, we have
\begin{equation}
\nonumber
|A_{1}|
\leq \frac{n^{2/3}\mu^{1/3}}{\OPT^{1/3}} h
\leq n^{2/3} \left(\frac{\mu}{\OPT}\right)^{1/3} \OPT.
\end{equation}
\end{proof}

\begin{lemma}
\label{lm - no forbidden edge in E'}
No forbidden edge is in $E'$.
\end{lemma}
\begin{proof}
If $e$ is a forbidden edge, its cost is zero.
Since vertices in $R$ are layered according to their distances to $s$,
any forbidden edge cannot lie in two different layers.
By definition, $E'$ is the set of edges between two consecutive layers.
So, there is no forbidden edge in $E'$.
\end{proof}

\begin{lemma}
\label{lm - structural lemma on multi-graphs with A, B}
For Algorithm \ref{alg - LstC on multi-graphs with A, B}, we have
$|E'| = O(n^{2/3}\mu^{1/3}\OPT^{2/3})$,
where $\mu$ is the maximum multiplicity of edges in $\tilde G$.
\end{lemma}
\begin{proof}
When the algorithm enters stage two, we have
$d(s, t) > \frac{n^{2/3}\mu^{1/3}}{\OPT^{1/3}}$.
Therefore, there are at least
$\frac{n^{2/3}\mu^{1/3}}{\OPT^{1/3}} + 1$ different $i$'s
such that $V_i \neq \emptyset$. We omit the vertices in $G$
that are not reachable from $s$.

We claim that there exists an $i' \leq \frac{n^{2/3}\mu^{1/3}}{\OPT^{1/3}}$
such that $|V_{i'}||V_{i'+1}|\leq100(\frac{n \cdot \OPT}{\mu})^{2/3}$.
Assume for contradiction for every
$0 \leq i \leq \frac{n^{2/3}\mu^{1/3}}{\OPT^{1/3}}$,
we have $|V_i||V_{i+1}| > 100(\frac{n \cdot \OPT}{\mu})^{2/3}$.
Then for each two consecutive layers $V_i$ and $V_{i+1}$, at least
one of them has size $> 10(\frac{n \cdot \OPT}{\mu})^{1/3}$.
Therefore, there are at least $\frac12 \cdot \frac{n^{2/3}\mu^{1/3}}{\OPT^{1/3}}$
different $i$'s satisfying $|V_i| > 10(\frac{n \cdot \OPT}{\mu})^{1/3}$,
and thus there are at least $5n$ vertices in these $V_i$'s.
This is obviously absurd since there are $\leq n$ vertices in $R$.

Therefore, the number of edges $|E_{i'}|$ between $V_{i'}$ and $V_{i'+1}$
is at most $|V_{i'}||V_{i'+1}|\mu = O(n^{2/3}\OPT^{2/3}\mu^{1/3})$,
considering the maximum multiplicity of edges is $\mu$.
The lemma follows since $E'$ is the minimum cardinality one among all $E_i$'s.
\end{proof}

\begin{theorem}
\label{th - LstC on multi-graphs with A, B}
For the unweighted {\sf Label $s$-$t$ Cut} problem on multi-graphs
with forbidden labels, if it has feasible solutions,
then Algorithm \ref{alg - LstC on multi-graphs with A, B}
returns an $O(\frac{n^{2/3}\mu^{1/3}}{\OPT^{1/3}})$-approximation
in polynomial time. Otherwise, the algorithm returns failure in polynomial
time.
\end{theorem}
\begin{proof}
First note that Algorithm \ref{alg - LstC on multi-graphs with A, B}
obviously runs in polynomial time.

Let $\cal I$ be the input instance.
If $\cal I$ has no feasible solution, then there is (at least)
one $s$-$t$ paths in $\tilde{G}$ whose edges are all forbidden edges.
That is, there is an $s$-$t$ path of length zero.
Such a path will be found at step \ref{step - failure} before the stage one.
So, in this case Algorithm \ref{alg - LstC on multi-graphs with A, B} will
return failure.

Suppose that $\tilde{\cal I}$ has feasible solutions.
The labels selected by Algorithm \ref{alg - LstC on multi-graphs with A, B}
in stage one are all admissible labels.
By Lemma \ref{lm - no forbidden edge in E'}, the labels selected
in stage two are also all admissible labels.
So, the algorithm returns a feasible solution in this case.

Let $SOL$ be the solution value of Algorithm
\ref{alg - LstC on multi-graphs with A, B}.
By Lemma \ref{lm - |A1| on multi-graphs with A, B} and
Lemma \ref{lm - structural lemma on multi-graphs with A, B},
we get that
\[
SOL = |A_{1} \cup A_2|
\leq \frac{n^{2/3}\mu^{1/3}}{\OPT^{1/3}} \OPT +
    n^{2/3}\OPT^{2/3}\mu^{1/3}
= O \left( \frac{n^{2/3}\mu^{1/3}}{\OPT^{1/3}} \right) \OPT,
\]
finishing the proof of the theorem.
\end{proof}

\section{Weighted Label $s$-$t$ Cut}
\label{sec - weighted LstC}
Recall that In the weighted {\sf Label $s$-$t$ Cut} problem,
we are given a (directed or undirected)
graph $G = (V, E)$, a source $s \in V$, a sink $t \in V$, and a label set
$L = \{\ell_1, \ell_2, \dots, \ell_q\}$, where each label $\ell \in L$ has
a nonnegative weight $w_\ell$. Each edge $e$ in graph $G$ has a label
$\ell(e) \in L$. The goal of the problem is to find a minimum weight
label subset $L' \subseteq L$ such that $s$ and $t$ are disconnected
if the edges with labels in $L'$ are removed from $G$.
Such a label subset $L'$ is also called a label $s$-$t$ cut.
Given a label subset $L'$, its weight $w(L')$ is the sum of weights of
all labels in it.

We show how to approximate the weighted {\sf Label $s$-$t$ Cut} problem
by reducing it to the (unweighted) {\sf Label $s$-$t$ Cut} problem
on multi-graphs with forbidden labels.

\subsection{The Reduction}
\label{sec - reduce weighted LstC to LstC on multigraphs}
The overall strategy of approximating weighted {\sf Label $s$-$t$ Cut}
is to reduce the weighted problem on (simple) graphs to the unweighted problem
on multi-graphs.
Let ${\cal I} = (G, s, t, L, w)$ be an instance of weighted {\sf Label $s$-$t$ Cut}.
The resulting instance of unweighted {\sf Label $s$-$t$ Cut} will be denoted
by $\tilde{\cal{I}} = (\tilde{G}, s, t, \tilde{A}, \tilde{B})$, where $\tilde{G}$
is a multi-graph, $\tilde A$ and $\tilde B$ are two sets of labels.
For an edge $e = (u, v) \in E(G)$, there will be multiple of edges
between $u$ and $v$ in $\tilde G$. In other words, in $\tilde G$ we will use
the multiplicity of edges between $u$ and $v$ to reflect the weight of
label $\ell(e)$.

However, the label weights of instance $\cal I$ are not polynomially
bounded in general. Since the reduction have to be done in polynomial time,
the crux in the reduction is to guarantee in $\tilde G$ the edge multiplicity
between any pair of vertices is polynomially bounded, meanwhile it is yet
used to reflect label weight (which may be not polynomially bounded).
To realize this seemingly contradictory goal,
we classify the labels in $L$ into two categories, namely,
the admissible labels and the forbidden labels (which are defined in
(\ref{eqn - A, set of admissible labels}) and
(\ref{eqn - B, set of forbidden labels})).
An admissible label $\ell$ in instance $\cal I$ will be translated into
a group of labels in instance $\tilde{\cal I}$.
The weight of $\ell$ is represented by the size of the group.
On the other hand, the forbidden labels will not be used in the solution
we constructed to instance $\cal I$. Therefore, we need not to translate
forbidden labels.

Specifically, let $O \subseteq L$ be an optimal solution to instance $\cal I$,
and $W$ be the maximum weight of labels in $O$,
i.e.,
\[
W = \max \{w_\ell \mid \ell \in O\}.
\]
By definition, it obviously holds that
\begin{equation}
\label{eqn - W <= OPT(I)}
W \leq \OPT_w({\cal I}).
\end{equation}
Here we use $\OPT_w({\cal I})$ to denote the optimum of
the weighted {\sf Label $s$-$t$ Cut} problem on instance $\cal I$.

By the guess technique, we may assume that we know the value $|O|$ and $W$,
where $|O|$ is the number of labels used in the optimal solution $O$.

Define
\begin{eqnarray}
\label{eqn - A, set of admissible labels}
A &=& \{\ell \in L \mid w_\ell \leq W\}, \\
\label{eqn - B, set of forbidden labels}
B &=& L \setminus A.
\end{eqnarray}
The labels in $A$ are called {\em admissible labels},
while the labels in $B$ are called {\em forbidden labels}.

For each admissible label $\ell \in A$, we define
\begin{equation}
\label{eqn - tilde w_l}
\bar{w}_\ell = \left\lceil \frac{w_\ell |O|}{W} \right\rceil.
\end{equation}
It is easy to see that $\bar{w}_\ell$ is an integer whose value is at least one.
Since $w_\ell \leq W$ for each $\ell \in A$, we know that $\bar{w}_\ell$
is at most $|O|$, which is polynomially bounded. This is a key property
about $\bar{w}_\ell$.

The guess technique is actually to try every possible values of
the quantity to be guessed. For the quantity $|O|$, we will try
its each possible values from $1$ to $q$. So, the guess for $|O|$
can be done in polynomial time and each guess of $|O|$ is polynomially
bounded. For the value $W$, we will try each distinct label weight
in $\{w_\ell \mid \ell \in L\}$. So, the guess for $W$ can also be done
in polynomial time.

Now we are ready to construct the instance
$\tilde{\cal{I}} = (\tilde{G}, s, t, \tilde{A}, \tilde{B})$.
Initially, the label set $\tilde{A}$ is empty, and graph $\tilde{G}$
contains all vertices in $V(G)$ and no edges.
Then, for each admissible label $\ell \in A$, we put $\bar{w}_\ell$ copies
of $\ell$ into $\tilde{A}$.
Let $\ell^{(1)}, \ell^{(2)}, \dots, \ell^{(\bar{w}_\ell)}$ be these copies.
For the sake of convenience, given an admissible label $\ell \in A$, we define
\[
g(\ell) = \{\ell^{(1)}, \ell^{(2)}, \dots, \ell^{(\bar{w}_\ell)}\}
\]
as the label group associated with $\ell$.
The label set $\tilde{B}$ is just equal to $B$. For the instance $\tilde{\cal I}$,
the labels in $\tilde{A}$ are admissible labels,
and the labels in $\tilde B$ are forbidden labels.

Let $e = (u, v)$ be an edge in graph $G$ with label $\ell = \ell(e)$.
If $\ell$ is an admissible label in $A$, we put $\bar{w}_\ell$ edges between $u$ and $v$
in graph $\tilde{G}$, with the $i$-th copy ($1 \leq i \leq \bar{w}_\ell$)
of the edge labeled with $\ell^{(i)}$. Otherwise ($\ell$ is a forbidden label in $B$),
we just put $e$ itself in $\tilde{G}$.
The construction of instance $\tilde{\cal{I}}$ is finished.
Note that $\tilde{\cal I}$ is just the instance $\cal I$
in Section \ref{sec - LstC on multi-graphs with forbidden labels}.

Here is the algorithm for weighted {\sf Label $s$-$t$ Cut}.


\begin{algz}
\label{alg - weighted LstC}
\setcounter{algleo}{0}
\begin{algleo}
\linonumber
\linooffset {\em Input:} An instance ${\cal I} = (G, s, t, L, w)$ of weighted
    {\sf Label $s$-$t$ Cut}.
\linooffset {\em Output:} A label subset $L' \subseteq L$.
\li Guess $|O|$ and $W$.
\li Construct an instance $\tilde{\cal I} = (\tilde{G}, s, t, \tilde{A}, \tilde{B})$
    of {\sf Label $s$-$t$ Cut} on multi-graphs as stated in
    Section \ref{sec - reduce weighted LstC to LstC on multigraphs}.
\li {\bf call} Algorithm \ref{alg - LstC on multi-graphs with A, B}
    on instance $\tilde{\cal I}$, obtaining a solution $\tilde{A}' \subseteq \tilde{A}$.
\li \label{step - convert tilde sol to sol}
    {\bf return} $L' \leftarrow \{\ell \in A \mid g(\ell) \subseteq \tilde{A}'\}$.
\end{algleo}
\end{algz}

Algorithm \ref{alg - weighted LstC} converts solution $\tilde{A}'$
to a label subset $L' \subseteq L$ in its step \ref{step - convert tilde sol to sol}
in the following way. Let $\ell^{(i)}$ be a label appeared in $\tilde{A}'$.
By the reduction, there is a label group $g(\ell) \subseteq \tilde{A}$
created from $\ell$, to which $\ell^{(i)}$ belongs. If $g(\ell)$ appears
{\em completely} in $\tilde{A}'$ (i.e., $g(\ell) \subseteq \tilde{A}'$),
we put label $\ell$ into $L'$. Note that $group_\ell$ may appear only
{\em partially} in $\tilde{A}'$. In this case we do not put label $\ell$ into $L'$.

\subsection{Analysis}
Let $\cal I$ be an instance of weighted {\sf Label $s$-$t$ Cut},
and $\tilde{\cal I}$ be the instance of {\sf Label $s$-$t$ Cut} on multi-graphs
created from $\cal I$ by Algorithm \ref{alg - weighted LstC}.
We now analyze the approximation ratio of Algorithm \ref{alg - weighted LstC}.
First we show the relationship between the optimum of instance $\tilde{\cal I}$,
denoted by $\OPT_m(\tilde{\cal I})$, and the optimum of instance $\cal I$,
denoted by $\OPT_w({\cal I})$. For clarity, we use subscript ``m''
in $\OPT_m()$ to specify it is the optimum of some instance of the
{\sf Label $s$-$t$ Cut} problem on multi-graphs with forbidden labels.
Similarly, we use subscript ``w'' in $\OPT_w()$ to specify it is the optimum
of some instance of the weighted {\sf Label $s$-$t$ Cut} problem.

\begin{lemma}
\label{lm - OPT(tilde I) <= func of OPT(I)}
For the correct guess of $|O|$ and $W$, we have
$\OPT_m(\tilde{\cal I}) \leq \frac{|O|}{W} \OPT_w({\cal I}) + |O|$.
\end{lemma}
\begin{proof}
By (\ref{eqn - tilde w_l}), we have $\bar{w}_\ell \leq \frac{w_\ell |O|}{W} + 1$,
i.e., $w_\ell \geq \frac{W}{|O|}(\bar{w}_\ell - 1)$. Then it follows that
\[
\OPT_w({\cal I}) = \sum_{\ell \in O} w_\ell
\geq \frac{W}{|O|} \sum_{\ell \in O} (\bar{w}_\ell - 1)
= \frac{W}{|O|} \left(\left(\sum_{\ell \in O} \bar{w}_\ell\right) - |O|\right).
\]
By replacing each label $\ell$ in $O$ with $\ell^{(1)}$, $\ell^{(2)}$, $\dots$,
$\ell^{(\bar{w}_\ell)}$ (note that all of them are in $\tilde A$,
the admissible label set of $\tilde{\cal I}$),
we can get a feasible solution to instance $\tilde{\cal I}$
with cost $\sum_{\ell \in O} \bar{w}_\ell$.
This means that $\sum_{\ell \in O} \bar{w}_\ell \geq \OPT_m(\tilde{\cal I})$.
So, we have
\[
\OPT_w({\cal I}) \geq \frac{W}{|O|} \left(\OPT_m(\tilde{\cal I}) - |O|\right),
\]
obtaining the lemma.
\end{proof}

\begin{lemma}
\label{lm - OPT(tilde I) >= func of OPT(I)}
For the correct guess of $|O|$ and $W$, we have
$\OPT_m(\tilde{\cal I}) \geq \frac{|O|}{W} \OPT_w({\cal I})$.
\end{lemma}
\begin{proof}
Since $W$ is a correct guess of the maximum weight of labels in $O$,
each label in $g(\ell)$ for every $\ell \in O$ is an admissible label
with respect to instance $\tilde{\cal I}$. By replacing each label $\ell$ in $O$
with $\ell^{(1)}, \ell^{(2)}, \dots, \ell^{(\bar{w}_\ell)}$,
we can get a feasible solution to instance $\tilde{\cal I}$.
That is, instance $\tilde{\cal I}$ has feasible solutions.

So, we may let $\tilde{O} \subseteq \tilde{A}$ be an optimal solution
to instance $\tilde{\cal I}$. From $\tilde{O}$, we can construct a feasible
solution $A'$ to instance $\cal I$, by replacing each label group
$\ell^{(1)}, \ell^{(2)}, \dots, \ell^{(\bar{w}_\ell)}$ in $\tilde{O}$
with label $\ell$.
By (\ref{eqn - tilde w_l}), we have $\bar{w}_\ell \geq \frac{w_\ell |O|}{W}$,
i.e., $w_\ell \leq \frac{W}{|O|}\bar{w}_\ell$.
Then it follows that
\[
\OPT_w({\cal I})
\leq w(A') = \sum_{\ell \in A'} w_\ell
\leq \frac{W}{|O|} \sum_{\ell \in A'} \bar{w}_\ell
= \frac{W}{|O|} |\tilde{O}|
= \frac{W}{|O|} \OPT_m(\tilde{\cal I}).
\]
Rearranging the terms gives the lemma.
\end{proof}

\begin{lemma}
\label{lm - weight of L'}
For the correct guess of $|O|$ and $W$, we have
\[
w(L')
= O\left(\frac{n^{2/3}\mu^{1/3}}{\OPT_m(\tilde{\cal I})^{1/3}}\right)
    \OPT_w({\cal I}).
\]
\end{lemma}
\begin{proof}
By (\ref{eqn - tilde w_l}), we have
$\bar{w}_\ell \geq \frac{w_\ell |O|}{W}$,
implying $w_\ell \leq \frac{W}{|O|}\bar{w}_\ell$.
So, we get
\begin{equation}
\label{eqn - w(L') part 1}
w(L') = \sum_{\ell \in L'} w_\ell
\leq \frac{W}{|O|} \sum_{\ell \in L'}\bar{w}_\ell.
\end{equation}

By the construction of solution $L'$, we know
$\sum_{\ell \in L'}\bar{w}_\ell \leq |\tilde{A}'|$.
By Theorem \ref{th - LstC on multi-graphs with A, B},
$|\tilde{A}'| \leq \rho(\tilde{\cal I}) \OPT(\tilde{\cal I})$,
where
$\rho(\tilde{\cal I}) =
    O(\frac{n^{2/3}\mu(\tilde{G})^{1/3}}
    {\OPT_m(\tilde{\cal I})^{1/3}})$,
$n$ is the number of vertices in $\tilde G$,
and $\mu(\tilde{G})$ is the the maximum edge multiplicity of $\tilde G$.
Note that the multi-graph $\tilde G$ has the same number of vertices
as in $G$. That is, we have
\begin{equation}
\label{eqn - w(L') part 2}
\sum_{\ell \in L'}\bar{w}_\ell
\leq \rho(\tilde{\cal I}) \OPT_m(\tilde{\cal I}).
\end{equation}

By (\ref{eqn - w(L') part 1}), (\ref{eqn - w(L') part 2}), and
Lemma \ref{lm - OPT(tilde I) <= func of OPT(I)}, we have
\begin{eqnarray*}
w(L') &\leq& \frac{W}{|O|} \rho(\tilde{\cal I})
    \left(\frac{|O|}{W} \OPT_w({\cal I}) + |O|\right) \\
&=& \rho(\tilde{\cal I}) \OPT_w({\cal I}) + \rho(\tilde{\cal I}) W \\
&\stackrel[\text{(\ref{eqn - W <= OPT(I)})}]{}{=}&
    O\left(\frac{n^{2/3}\mu(\tilde{G})^{1/3}}{\OPT_m(\tilde{\cal I})^{1/3}}\right)
    \OPT_w({\cal I}).
\end{eqnarray*}
\end{proof}

\begin{theorem}
The weighted {\sf Label $s$-$t$ Cut} problem can be approximated within $O(n^{2/3})$.
\end{theorem}
\begin{proof}
By (\ref{eqn - tilde w_l}), the multiplicity $\mu(\tilde G)$ of multi-graph
$\tilde G$ is at most $|O|$. So, by Lemma \ref{lm - weight of L'}, the approximation ratio
of Algorithm \ref{alg - weighted LstC} is
\begin{eqnarray*}
n^{2/3}\left(\frac{\mu(\tilde{G})}{\OPT_m(\tilde{\cal I})}\right)^{1/3}
&\stackrel[\text{(\ref{eqn - tilde w_l})}]{}{\leq}&
    n^{2/3}\left(\frac{|O|}{\OPT_m(\tilde{\cal I})}\right)^{1/3} \\
&\stackrel[\text{LM\ref{lm - OPT(tilde I) >= func of OPT(I)}}]{}{\leq}&
n^{2/3}\left(\frac{W}{\OPT_w({\cal I})}\right)^{1/3}
\stackrel[\text{(\ref{eqn - W <= OPT(I)})}]{}{\leq}
    n^{2/3}.
\end{eqnarray*}

Finally, it is straightforward that Algorithm \ref{alg - weighted LstC}
runs in polynomial time. The theorem follows.
\end{proof}

\section{Concluding Remarks}
\label{sec - conclusions}
In this paper, we design an approximation algorithm for the weighted
{\sf Label $s$-$t$ Cut} problem whose ratio is $O(n^{2/3})$, where
$n$ is the number of vertices in the input graph.
An interesting open problem for {\sf Label $s$-$t$ Cut} is whether
the exponent of the ratio $O(n^{2/3})$ can be improved.
For example, is an $O(n^{1/2})$-approximation for (weighted or unweighted)
{\sf Label $s$-$t$ Cut} possible?
Note that the current best approximation hardness factor for unweighted
{\sf Label $s$-$t$ Cut} is $2^{(\log n)^{1-1/(\log\log n)^c}}$
for any constant $c < 1/2$ \cite{ZCTZ11}.

The approximation result $O(n^{2/3})$ for weighted {\sf Label $s$-$t$ Cut}
applies to the weighted {\sf Global Label Cut} problem as well.
For the {\sf Global Label Cut} problem, a long-term open problem is whether
it is in P or NP-hard.

\section*{Acknowledgements}
Peng Zhang is supported by the National Natural Science Foundation of China
(61972228, 61672323, 61672328), and the Natural Science Foundation of Shandong Province,
China (ZR2019MF072, ZR2016AM28).



\bibliographystyle{plain}
\bibliography{lcbib}
\end{document}